\documentclass[aps,twocolumn,superscriptaddress,amsmath,amssymb,10pt,floatfix,nofootinbib]{revtex4}
\usepackage{amsthm,graphicx,mathrsfs,bm}
\usepackage[usenames,dvipsnames,svgnames,table]{xcolor}
\usepackage[unicode=true,pdfusetitle,
 bookmarks=true,bookmarksnumbered=false,bookmarksopen=false,
 breaklinks=true,pdfborder={0 0 0},backref=false,colorlinks=true]
 {hyperref}
\hypersetup{
 linkcolor=NavyBlue,urlcolor=NavyBlue,citecolor=NavyBlue}

\newcommand{\bra}[1]{{\langle{#1}|}}
\newcommand{\ket}[1]{{|{#1}\rangle}}
\newcommand{\braket}[1]{\mathinner{\langle{#1}\rangle}}
\newcommand{\cC}{\mathcal{C}}

\newcommand{\cK}{\mathcal{K}}
\newcommand{\cL}{\mathcal{L}} 
\newcommand{\cN}{\mathcal{N}}

\DeclareMathOperator{\tr}{Tr} 
\newtheorem{theorem}{Theorem}
\newtheorem{lemma}{Lemma}

\begin{document}

\title{Robust quantum metrological schemes based on protection of quantum Fisher information}

\author{Xiao-Ming Lu}
\affiliation{Centre for Quantum Technologies, National University of Singapore, 3 Science Drive 2, Singapore 117543, Singapore}
\affiliation{Department of Electrical and Computer Engineering, National University of Singapore, 4 Engineering Drive 3, Singapore 117583, Singapore}

\author{Sixia Yu} 
\affiliation{Centre for Quantum Technologies, National University of Singapore, 3 Science Drive 2, Singapore 117543, Singapore}
\affiliation{Hefei National Laboratory for Physical Sciences at Microscale and Department of Modern Physics, University of Science and Technology of China, Hefei, Anhui 230026, China}
\author{C.H. Oh}
\affiliation{Centre for Quantum Technologies, National University of Singapore, 3 Science Drive 2, Singapore 117543, Singapore}
\affiliation{Department of Physics, National University of Singapore, 3 Science Drive 2, Singapore 117543, Singapore}

\begin{abstract} 
Fragile quantum features such as entanglement are employed to improve the precision of parameter estimation and as a consequence the quantum gain becomes vulnerable to noise. As an established tool to subdue noise, quantum error correction is unfortunately overprotective because the quantum enhancement can still be achieved even if the states are irrecoverably affected, provided that the quantum Fisher information, which sets the ultimate limit to the precision of metrological schemes, is preserved and attained. 
Here, we develop a theory of robust metrological schemes that preserve the quantum Fisher information instead of the quantum states themselves against noise. After deriving a minimal set of testable conditions on this kind of robustness, we construct a family of $2t+1$ qubits metrological schemes being immune to $t$-qubit errors after the signal sensing. In comparison at least five qubits are required for correcting arbitrary 1-qubit errors in standard quantum error correction.
\end{abstract}

\maketitle

In quantum metrology, delicate and fragile quantum features are being used to enhance the sensitivity of experimental apparatus, e.g., non-classical probe states were used for the high sensitivity of optical interferometer and atomic spectroscopy~\cite{Caves1981,Yurke1986,Wineland1992,Holland1993,Dowling1998,Giovannetti2006,Giovannetti2004,Giovannetti2011}. 
However the quantum enhancement for the sensitivity may be subdued by the presence of ubiquitous and inevitable noise~\cite{Huelga1997,Escher2011,Chaves2013,Demkowicz-Dobrzanski2012,Tsang2013,Dorner2009,Rubin2007,Huver2008,Lee2009,Maccone2009,Ono2010,Joo2011,Jiang2012,Spagnolo2012,Kacprowicz2010,Demkowicz-Dobrzanski2009,Genoni2011,Genoni2012}.
Therefore it is of utmost significance to investigate the robustness of the optimal strategies for the high sensitivity against noise.
Quantum error correction (QEC) was employed in quantum metrology~\cite{Macchiavello2002,Preskill2000,Arrad2014,Kessler2014,Ozeri2013,Dur2014} to overcome noise problem, where by protecting the quantum states, on which a signal parameter is imprinted, the measurement precision for that parameter is protected.

This is no wonder because the standard QEC was originally designed for protecting all the information encoded in quantum states, i.e., the logical states in the universal quantum computation, against the noise~\cite{Shor1995,Bennett1996,Steane1996,Gottesman1996,Knill1997,Yu2008,Yu2013}.
In quantum metrology, however, what matters essentially is the distinguishability about the signal parameter that is sensed by quantum systems and encoded in quantum states. According to quantum estimation theory~\cite{Helstrom1976,Holevo1982,Braunstein1994,WisemanBook,Paris2009} this distinguishability is measured by quantum Fisher information (QFI). Therefore, preserving the QFI of a given family of states against noise is sufficient for quantum metrological schemes to work under noisy environment.
Since the QFI represents only partial information encoded in quantum states, the use of the  QEC for quantum states is obviously overprotective, which leads to unnecessary waste of resources.
Our main goal is to establish a variant theory of QEC designed for quantum metrology, namely, robust quantum metrological schemes, by taking the QFI instead of the fidelity of quantum states as the figure of merit.

In this paper, we show that analogous to QEC for quantum states the errors can also be digitalized so that we can construct a robust metrological scheme, in which the QFI is preserved under an entire class of unknown noisy processes rather than a specific one. 
Furthermore, we derive the necessary and sufficient testable conditions on preserving QFI, and construct the optimal measurements extracting the maximal distinguishability about the signal parameter in the presence of noise.
Our testable conditions describe the minimal requirements for the robustness of a parameter estimation scheme against noise, and can be used to identify the errors to which the QFI is immune.
As an example, we construct a family of metrological scheme on $2t+1$ physical qubits to protect the QFI against arbitrary errors on no more than $t$ physical qubits after the signal sensing.

\begin{figure*}[tb]
\begin{center}
 \includegraphics[width=180mm]{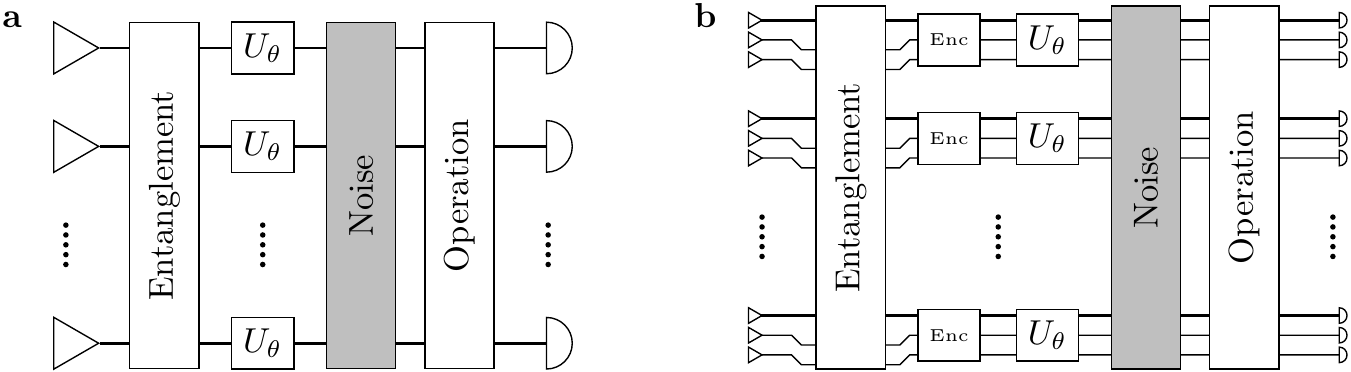}
\end{center}
\caption{\label{fig:model}
{\bf Abstract models for quantum parameter estimation.}
{\bf a}, Set-up for entanglement-enhanced metrology with noise being assumed after the sensing transformation $U_\theta^{\otimes n}$, where $n$ is the number of the qubits. 
{\bf b}, Heisenberg-limited metrological scheme with the QEC protection. 
The entangling operation (labeled with ``Entanglement'' in the figure) acts only on the first qubit (thick line) in each block, and produce the GHZ state.
The encoders (labeled with ``Enc'') encode the first qubit in each block into a phase-flip code. 
The sensing transformation $U_\theta$ is a logical phase-shift operation on each phase-flip code space.
Here, we note that the ``Entanglement'' operation plays a dual role: on one hand, it supplies the Heisenberg-scaling of the QFI, and on the other hand, it supplies a higher level of bit-flip code  in addition to the phase-flip code. }
\end{figure*}

\vspace{8pt}
\noindent\textbf{Results}\\
{\bf Quantum parameter estimation theory.}
A standard quantum metrological scheme to detect and estimate a signal parameter $\theta$ can be depicted by the following sensing transformation: 
\begin{equation}
	\rho\mapsto \rho_\theta=e^{-i\theta H}\rho e^{i\theta H}
\end{equation}
with $H$ being a known Hermitian operator and $\rho$ the probe state.
The value of the parameter is estimated through the classical data processing on the measurement outcomes obtained by repeating experiments in $\rho_\theta$.
From estimation theory~\cite{Helstrom1976,Holevo1982,Braunstein1994}, the regularized root-mean-square error of the estimator $\hat\theta$ is limited by the Cram\'er-Rao bound 
\begin{equation}
    \delta\hat{\theta}:=\Big\langle\Big(\frac{\hat\theta}{|\mathrm{d}\braket{\hat\theta}/\mathrm{d}\theta|}-\theta\Big)^2\Big\rangle^{1/2}\geq \frac 1{\sqrt{\nu F(\rho_\theta|M)}},	
\end{equation}
where $\nu$ is the number of repetitions of the experiments, and 
\begin{equation}
	F(\rho_\theta|M):=\sum_x p_\theta(x)\Big[\frac{\mathrm{d}}{\mathrm{d}\theta}\ln p_\theta(x)\Big]^2
\end{equation}
is the (classical) Fisher information extracted by the measurement $\{M_x\}$ with $p_\theta(x):= \tr(M_x\rho_\theta)$ being the probabilities of obtaining outcomes $x$. 
Here, $M_x$ are positive operators satisfying $\sum_x M_x=\bm{1}$ with $\bm{1}$ being the identity operator.
The maximal Fisher information over all possible measurements is given by the so-called QFI $F(\rho_\theta ):=\tr(\rho_\theta L_\theta ^{2}),$ where the symmetric logarithmic derivative (SLD) operator $L_\theta $ is defined as the Hermitian operator satisfying ${\mathrm{d}\rho_\theta}/{\mathrm{d}\theta}=\frac12\{L_\theta,\rho_\theta\}$ with $\{\cdot,\cdot\}$ being the anti-commutator~\cite{Braunstein1994,Helstrom1976,Holevo1982,WisemanBook,Paris2009}.
More importantly, the Cram\'er-Rao bound is asymptotically achieved~\cite{Helstrom1976,Holevo1982}, therefore, QFI can be considered as a measure on the distinguishability about the parameter in quantum states.

An optimal strategy of the quantum parameter estimation comprises the probe state maximizing the QFI and the measurement attaining the maximal Fisher information.
Taking appropriate entangled states as the probe states, a quantum metrological scheme may achieve the Heisenberg scaling of precision $1/N$, where $N$ is the number of resources employed in the experiment, e.g., the number of probes~\cite{Giovannetti2006,Giovannetti2004,Giovannetti2011}.
This is a considerable improvement over the standard quantum limit $1/\sqrt{N}$.
Nevertheless, those entanglement-enhanced strategies that are optimal for the noiseless systems easily lose the quantum gain for the noisy systems~\cite{Huelga1997,Chaves2013,Demkowicz-Dobrzanski2012,Dorner2009,Rubin2007,Huver2008,Lee2009,Maccone2009,Ono2010,Joo2011,Jiang2012,Spagnolo2012,Kacprowicz2010,Demkowicz-Dobrzanski2009,Genoni2011,Genoni2012}.

{\bf Parameter estimation in noisy cases.}
We now turn to the question of the optimal strategy in noisy cases.
We assume that the noise can be deferred until after the sensing transformation, i.e., states to be measured are $\cN(\rho_\theta)=\sum_j E_j\rho_\theta E_j^\dagger$ where $\cN$ denotes a noisy channel with Kraus operators $\{E_j\}$.
We emphasize that this noise model is applicable to the noise that commutes with the generator of the sensing transformation, occurs during the transmission or storage in the interval between the sensing and the measurement, or is induced by the measurement imperfection. 
This noise model can also be considered as an approximation when the sensing time is short~\cite{Dur2014,Kessler2014}.
A general entanglement-enhanced metrology scenario of this type is depicted in Fig~\ref{fig:model}a.
At first glance, the optimal strategy for such noisy cases might be established by seeking the optimal probe states maximizing the QFI of $\cN(\rho_\theta)$ and the corresponding optimal measurements.
Technically, this straightforward optimization needs to diagonalize the parametric family of states $\cN(\rho_\theta)$, which is often formidable and even impossible without the details of the noise.
Therefore, the optimal strategies obtained in this way are very restricted.

Based on the above considerations, protecting the involved parametric family of states with quantum error-correcting codes~\cite{Shor1995,Bennett1996,Steane1996,Gottesman1996,Knill1997}, which are applicable for the whole class of noisy channels with the Kraus operators being arbitrary linear combinations of the correctable error elements, is a good candidate of a robust strategy for quantum metrology~\cite{Arrad2014,Kessler2014,Ozeri2013,Dur2014}.
However, QEC is overprotective because the measurement precision remains the same as long as the QFI is preserved and attained, even if the quantum states might be affected by some uncorrectable errors.
Here we shall develop a theory of QEC specialized for metrology: on the one hand our theory aims at preserving the QFI instead of all the information encoded in states, which ensures that our robust quantum metrological schemes are not overprotective; on the other hand our specialized theory also possesses the great advantage of the standard QEC that the errors can be digitalized for the preservation and attainment of QFI, which is our first main result:

\begin{theorem}\label{thm:discretization}
The QFI of $\rho_\theta$ is preserved under a known channel $\cN$ with Kraus operators $\{E_j\}$ if and only if 
\begin{equation}\label{eq:NS}
    \cL_\theta E_j\sqrt{\rho_\theta} = E_j L_\theta\sqrt{\rho_\theta} \quad (\forall\, j)
\end{equation}
with $\cL_\theta$ being the SLD operator for $\cN(\rho_\theta)$.
If the QFI of $\rho_\theta$ is preserved under a known noisy channel $\cN$, then it is preserved under all noisy channels whose Kraus operators are arbitrary linear combinations of $\{E_j\}$, with the optimal measurement being the eigenstates of $\cL_\theta$.
\end{theorem}

We split the proof of Theorem~\ref{thm:discretization} into three parts.
First, we prove in the Methods that equation (1) is the necessary and sufficient condition on QFI-preserving under a known noisy channel.
Second, for an arbitrary noisy channel $\cK$ with Kraus operators $\{K_j=\sum_ic_{ij}E_j\}$ being unknown linear combinations of $\{E_j\}$, we note that equation (\ref{eq:NS}) still holds with $E_j$ being replaced by $K_j$, therefore the QFI is preserved under the unknown noisy channel $\cK$. 
Third, it is known that the complete set of the eigenstates of the SLD operator is the optimal measurement basis attaining the maximal Fisher information~\cite{Braunstein1994}. 
Through equation (\ref{eq:NS}), the SLD operator $\cL_\theta$ for $\cN(\rho_\theta)$ can be readily checked to be also the SLD operator for $\cK(\rho_\theta)$. 
Therefore, the measurement with respect to eigenstates of $\cL_\theta$ is also optimal for $\cK(\rho_\theta)$. 
Remarkably, if the measurement basis is fixed, then a recovery operation is demanded to transform the basis of the optimal measurement to the fixed one; otherwise, we only need to perform the optimal measurement for noisy states, and no recovery operation is demanded.

\begin{figure}[tb]
\begin{center}
\includegraphics[width=85mm]{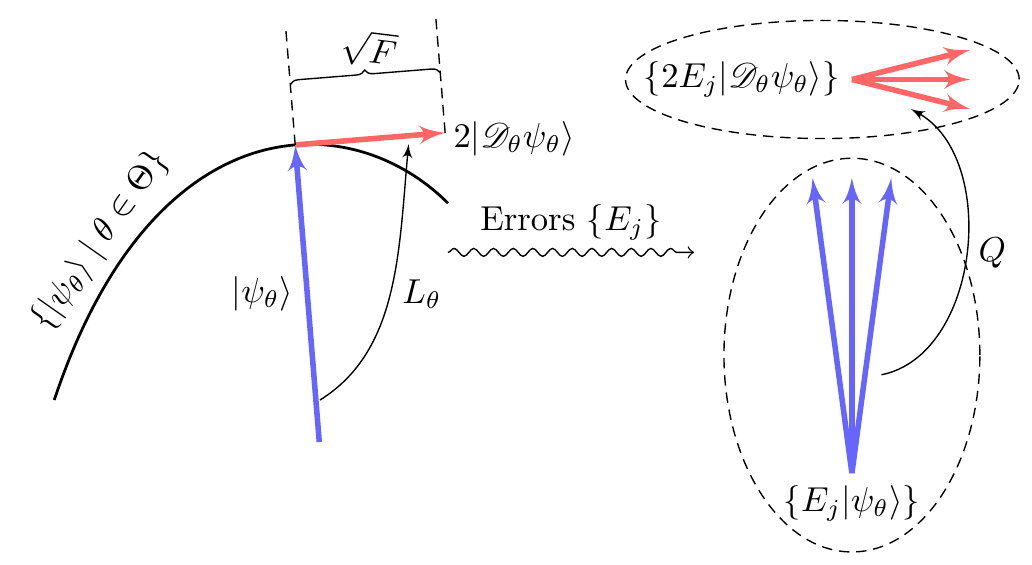}
\end{center}
\caption{\label{fig:geo}
{\bf Geometrical picture for preserving QFI.}
The QFI $F$ of $|\psi_\theta\rangle$ equals to the square of the Euclidean length of $2|\mathscr{D}_\theta\psi\rangle$ in terms of the Fubini-Study metric on the manifold of pure states, where $\ket{\mathscr{D}_\theta\psi_\theta}:=(\bm{1}-|\psi_\theta\rangle\langle\psi_\theta|)\frac{\mathrm{d}}{\mathrm{d}\theta}|\psi_\theta\rangle$ is the covariant derivative.
For pure states, the symmetric logarithmic derivative operator $L_\theta$ is a Hermitian representation of the covariant derivative as $L_\theta\ket{\psi_\theta}=2\ket{\mathscr{D}_\theta\psi_\theta}$, and the projective measurement with respect to eigenstates of $L_\theta$ attains the maximal Fisher information.
Under a set $\{E_j\}$ of errors, parametric state vectors are transformed into $E_j\ket{\psi_\theta}$, while covariant derivative vectors into $E_j\ket{\mathscr{D}_\theta\psi_\theta}$.
The QFI is preserved under the errors if and only if there exists a Hermitian operator $Q$ transforming all the erroneous state vectors to the corresponding erroneous covariant derivative vectors, i.e., $Q E_j\ket{\psi_\theta}=2E_j\ket{\mathscr{D}_\theta\psi_\theta}$ for all $j$; this operator $Q$ actually is the symmetric logarithmic derivative operator for all the noisy states under the errors (See the Methods). 
The projective measurement with respect to eigenstates of $Q$ attains the maximal Fisher information in noisy states.
}
\end{figure}

Theorem~\ref{thm:discretization} can be understood in a geometric way.
On the manifold of pure states, there is a Riemannian metric known as the Fubini-Study metric~\cite{Provost1980,Shapere1989}.
Along the parametric states $|\psi_\theta\rangle$, a line element is given by $\mathrm{d}s(|\psi_\theta\rangle,|\psi_{\theta+\mathrm{d}\theta}\rangle)=\lVert|\mathscr{D}_\theta\psi_\theta\rangle\rVert\mathrm{d}\theta$ with $\ket{\mathscr{D}_\theta\psi_\theta}:=(\bm{1}-|\psi_\theta\rangle\langle\psi_\theta|)\frac{\mathrm{d}}{\mathrm{d}\theta}|\psi_\theta\rangle$ being the covariant derivative vector.
This geometric metric is connected to parameter estimation theory in twofold:
on one hand, the QFI is given by $F(\ket{\psi_\theta}\bra{\psi_\theta}) = 4\lVert|\mathscr{D}_\theta\psi_\theta\rangle\rVert^2$;
on the other hand, the SLD operator $L_\theta$ is a Hermitian representation of the covariant derivative as $L_\theta\ket{\psi_\theta}=2\ket{\mathscr{D}_\theta\psi_\theta}$, and the projective measurement with respect to eigenstates of $L_\theta$ attains the maximal Fisher information.
In Fig.~\ref{fig:geo}, we show that how the conditions on the preservation of QFI can be intuitively understood in this geometric picture.

Theorem~\ref{thm:discretization} concerns the robustness of quantum parameter estimation with respect to noise.
It suggests that there might be a probe state which does not maximize the QFI under a specific noisy channel but ensures the QFI to be preserved and attained under an entire class of noisy channels. 
The necessary and sufficient condition (\ref{eq:NS}) on preserving the QFI against a set of errors needs the SLD operator $\cL_\theta$ for the corresponding noisy state $\cN(\rho_\theta)$. 
Our second main result is the testable conditions on preserving QFI without referring to the SLD operators of the noisy state.
These testable conditions are useful for finding good probe states for certain errors, or identifying those errors to which the QFI with certain probe state is immune.

\begin{theorem}\label{thm:testable}
The QFI of $\ket{\psi_\theta}$ is preserved under a set $\{E_j\}$ of errors, if and only if \textup{(i)}
\begin{equation}\label{eq:testable_condition}
    \braket{\psi_\theta|E_j^\dagger E_k L_\theta|\psi_\theta} = \braket{\psi_\theta|L_\theta E_j^\dagger E_k|\psi_\theta}
\end{equation}
for all $j$ and $k$, and \textup{(ii)} $\sum_j \alpha_j E_j \ket{\psi_\theta}=0$ for some  $\alpha_j\in\mathbb{C}$ infers $\sum_j \alpha_j E_j L_\theta\ket{\psi_\theta}=0$. For mixed state $\rho_\theta$ the QFI is preserved if and only if the above two conditions hold for all the states $|\psi_\theta\rangle$ in the range of $\rho_\theta$.
\end{theorem}

The proof is sketched in the Methods (a full version is deferred to the Supplementary Note 1).
For a unitarily parameterized family of pure states, $\ket{\psi_\theta}=\exp(-i\theta H)\ket{\psi}$, by noting $L_\theta\ket{\psi_\theta} = -2i\Delta H \ket{\psi_\theta}$ with $\Delta H:=H-\braket{\psi|H|\psi}$, we simplify the two testable conditions into (i) 
\begin{equation}
	\langle\psi_\theta|\{E_j^\dagger E_k,\Delta H\}|\psi_\theta\rangle=0
\end{equation}
for all $j$ and $k$ and (ii) $\sum_j \alpha_j E_j \ket{\psi_\theta}=0$ for some  $\alpha_j\in\mathbb{C}$ infers $\sum_j \alpha_j E_j H\ket{\psi_\theta}=0$.
The testable conditions describe the minimal requirements for the robustness of a parameter estimation scheme against noise, and are looser than that of QEC for the parametric family of states.
Recall that a set $\{E_j\}$ of errors is correctable for a code space if and only if 
\begin{equation}\label{eq:qec}
    \langle\phi|E_j^\dagger E_k|\varphi\rangle=0, \quad
    \langle\phi|E_j^\dagger E_k|\phi\rangle=\langle\varphi|E_j^\dagger E_k|\varphi\rangle 
\end{equation}
are satisfied for all $j,k$ and all pairs of orthonormal state vectors $|\phi\rangle$ and $|\varphi\rangle$ in the code space~\cite{Knill1997}.
Let us choose $|\psi\rangle$ and $H$ to be in a standard quantum error-correcting code so that $|\phi\rangle\propto|\psi_\theta\rangle$ and $|\varphi\rangle\propto L_\theta|\psi_\theta\rangle$ are two orthonormal states in the coding subspace. 
In such case, the first and second testable conditions are implied by the first and second equalities in equation (\ref{eq:qec}) respectively.

Henceforth, we simply say that $\rho_\theta$ is a robust metrological scheme with respect to a set $\{E_j\}$ of errors if the QFI of $\rho_\theta$ is preserved under $\{E_j\}$.
We show below that concrete robust metrological schemes can be easily constructed based on the stabilizer formalism~\cite{Gottesman1996}.
A stabilizer code $\cC(S)$ is the joint $+1$ eigenspace of the stabilizer group $S$, which is an Abelian subgroup of the $n$-qubit Pauli group, i.e., $S_i\ket{\psi} = \ket{\psi}$ for all $S_i\in S$ and all $\ket{\psi}\in\cC(S)$. 
A set $\{E_j\}$ of Pauli errors---$E_j$ are also elements of the $n$-qubit Pauli group---are correctable for this stabilizer code, if each  $E_j^\dagger E_k$ is either in the stabilizer group, or detectable, i.e., anticommutes with at least one element of the stabilizer group~\cite{Gottesman1996}.

\begin{theorem}\label{thm:qec}
In a metrological scheme $|\psi_\theta\rangle=e^{-i\theta H}|\psi\rangle$ where the probe state $|\psi\rangle$ is taken from the coding subspace of a stabilizer code $\cC(S)$ capable of correcting errors $\{E_j\}$ and $[H,S]=0$, the QFI is also immune to the errors $\{E_j\bar X\}$, where $\bar X$ is a Pauli error that commutes with $S$ while anticommutes with $\Delta H_\mathrm{eff}:= PHP-\langle\psi|H|\psi\rangle$ with $P$ being the projection onto $\cC(S)$. If the coding subspace is two-dimensional then the optimal measurement is the joint measurement of $S$ and $\bar X$.
\end{theorem}

The proof is sketched in the Methods (see Supplementary Note 1 for a full proof). 
Theorem~\ref{thm:qec} can be easily used to identify the QFI-immune error set for a given scheme.
As an example, we consider a system composed of $n=2t+1$ qubits that are labeled with the index set $I = \{1,2,\cdots,n\}$.
Let us denote  $X_i$, $Y_i$ and $Z_i$ the tensor products of the Pauli matrices $X$, $Y$, and $Z$ on the $i$th qubit and identity operators on other qubits, respectively, and $O_\alpha = \prod_{i\in\alpha} O_i$ with $\alpha\subseteq I$ for $O=X,Y,$ and $Z$.
Let  $\cC$ be the 2-dimensional subspace stabilized by $\{X_\alpha\mid|\alpha|=\mathrm{even}\}$, which is exactly the coding subspace of a stabilizer code capable of correcting all $t$-qubit phase-flip errors $\{Z_\alpha\}$ with $|\alpha|\le t$. 
For any state $|\psi\rangle\in\cC$ such that $\langle\psi|Z_I|\psi\rangle=0$, the metrological scheme $\ket{\psi_\theta}=\exp(-i\theta Z_I)|\psi\rangle$ preserves QFI against all $t$-qubit phase flip errors plus errors of type $\{Z_\alpha X_I\}$, which include essentially arbitrary error on no more than $t$ qubits, i.e., those whose error operators have nontrivial effects on no more than $t$ qubits. 
That is to say, in terms of QFI, the $t$-qubit phase-flip codes can be used to protect a metrological scheme from arbitrary $t$-qubit errors occurring after the signal sensing. 
In comparison, at least five physical qubits are required in a standard quantum error-correcting code to correct arbitrary single qubit error, while our scheme requires only three physical qubits.
This is one of the advantages brought in by considering the preservation of QFI instead of the protection of quantum states. 
The maximal Fisher information is attained by the joint measurement of the stabilizers of $\cC$ and $X_I$, i.e., all the observables $\{X_j|j\in I\}$ without any recovery operation.

\vspace{8pt}
{\bf Entanglement-enhanced metrology.}
Beating the standard quantum limit by quantum entanglement is one of the most fascinating aspects of the quantum-enhanced metrology~\cite{Giovannetti2006,Giovannetti2011,Giovannetti2004}. 
A canonical example is utilizing the $m$-qubit Greenberger-Horne-Zeilinger (GHZ) state as the probe state for the parallel samplings of a unitary sensing transformation, wherein QFI scales quadratically with $m$---the Heisenberg scaling.
Replacing the noisy individual systems in the entangled state by logical ones makes the resulting scheme robust to correctable errors~\cite{Dur2014}.
Here, we show that the entanglement, besides helps to beat the standard quantum limit, also supplies a higher level of quantum error correcting code. 
Let us consider a metrological scheme whose parametric family of states read 
\begin{equation}\label{eq:HeisenbergLimit}
\ket{\psi_\theta}=\exp(-i\theta\sum_{i=1}^m\bar Z^{[i]})(\ket{\bar0}^{\otimes m}+\ket{\bar1}^{\otimes m})/\sqrt2,	
\end{equation}
where $\bar Z^{[i]}=\prod_{j=1}^n Z^{[i]}_j$ is the logical Pauli $Z$ operators on the $i$-th block, and $\ket{\bar0}$ and $\ket{\bar1}$ are the logical basis.
Let $S^{[i]} = \{X_\alpha^{[i]}\mid|\alpha|\text{ is even}\}$ be the stabilizer group of the $n$-qubit phase-flip code for the $i$-th qubit in the original scheme.
Further, assume that $m$ and $n=2t+1$ are odd.
This scheme is robust against to three kinds of errors occurring after the signal sensing.
First, less than or equal to $t$ phase-flip errors are correctable by phase-flip code in each block. 
Second, the states given by equation (\ref{eq:HeisenbergLimit}) are in a subspace stabilized by $\bar Z^{[i]}\bar Z^{[i+1]}$, which is a bit-flip code capable of correcting no more than $(m-1)/2$ logical bit-flip errors in the blocks. 
Since every single-qubit bit-flip error on the codewords of the phase-flip code in each block is equivalent to a logical bit-flip error on the block, less than or equal to $(m-1)/2$ physical bit-flip errors are correctable. 
Third, for more than $(m-1)/2$ bit-flip errors, the parametric family of states cannot be recovered but the QFI is still preserved.
Moreover, the joint measurement of all the stabilizers $X^{[i]}_j X^{[i]}_{j+1}$ and $\bar Z^{[i]}\bar Z^{[i+1]}$ of the stabilizer code together with $\prod_{ij}X_j^{[i]}$ attain the maximal Fisher information. 
In ref.~\cite{Dur2014}, D\"ur {\it et al}. proposed the same metrological scheme as equation (\ref{eq:HeisenbergLimit}), but only the protective capability of the error-correcting codes in each block was explored; we note that the GHZ state itself provides a higher-level bit-flip code and find some uncorrectable errors that are harmless to QFI.

\vspace{8pt}
{\bf Noise during the signal accumulation.} 
The above results still hold for the noise during the signal accumulation if the generator of the noise commutes with that of the signal accumulation.
A simple case of this kind is that the error operators $E_j$ commute with the generating operator $H$ of the signal accumulation;
but in such a case, our method is equivalent to the protection of quantum states~\cite{Dur2014} as the additional errors $\{E_j\bar X\}$ given by theorem~3 do not commute with $H$.
However, the errors that can be deferred after the signal accumulation are not restricted in this case.
Therefore, our method still has potential advantages over the protection of quantum states for the deferrable noise.
Let us consider a quantum system that evolves as 
\begin{equation}\label{eq:evolution}
\frac{\mathrm{d}\rho(t)}{\mathrm{d}t} = \omega\hat C\rho(t) + \hat D\rho(t),
\end{equation}
where $\omega$ is the signal parameter to be sensed and estimated, the superoperators $\hat C$ and $\hat D$ are the generators of the signal accumulation and the noise respectively. 
Usually, $\hat C\rho(t)=-i[H,\rho(t)]$ so that for the noiseless case  the signal accumulation is unitary.
The noise can be deferred after the signal accumulation if the superoperator $\hat C$ commutes with $\hat D$, so that the total evolution is $\exp(t\hat D)\exp(\omega t\hat C)$. 
This does not imply that the Kraus operators for $\exp(t\hat D)$ commute with the operator $H$.
For example, let us consider the phase accumulation of an atom under spontaneous emission, where $H=Z$ and $\hat D\rho=\gamma(\sigma_-\rho\sigma_-^\dagger- \sigma_-^\dagger\sigma_-\rho/2 - \rho\sigma_-^\dagger\sigma_-/2)$ with $\sigma_-=(X-iY)/2$.
The Kraus operators for $\exp(t\hat D)$ are given by $E_1=\frac{\sqrt\eta+1}{2}I+\frac{\sqrt\eta-1}{2}Z$ and $E_2=\sqrt{1-\eta}\sigma_-$ with $\eta=e^{-\gamma t}$.
It can be shown that $\hat D$ commutes with $\hat C$, nevertheless $[E_2,H]\neq0$.

\vspace{8pt}
{\bf Physical example.}
Here, we give a physical example to quantitatively analyze the performance of robust metrological schemes with the QFI-protection. 
Let us consider the frequency estimation of atoms with uncorrelated parallel and transverse dephasing. 
The atoms are modeled by qubits whose evolution is still in the form of equation (\ref{eq:evolution}). The noise is described by $\hat D=(\gamma_\mathrm{x}\hat D_\mathrm{x}+\gamma_\mathrm{z}\hat D_\mathrm{z})/2$, where $\gamma_\mathrm{x}$ and $\gamma_\mathrm{z}$ are the strengths of the noise, $\hat D_\mathrm{x}\rho=\sum_{i=1}^N (X_i\rho X_i-\rho)$, and $\hat D_\mathrm{z}\rho=\sum_{i=1}^N (Z_i\rho Z_i-\rho)$ with $N$ being the total number of qubits.
The qubits are divided into $m=\lfloor N/n\rfloor$ blocks.
The probe state is the logical GHZ states with respect to $n$-qubit phase-flip code in each block, and the generating operator of the signal accumulation is given by $H=\frac12\sum_{i=1}^m \bar Z^{[i]}$.
This scheme was first proposed by D\"ur {\it et al.}~\cite{Dur2014} to subdue only the parallel dephasing.
When $n=1$, this scheme is reduced to the ordinary one of using the raw GHZ probe state and independent signal accumulation~\cite{Giovannetti2004}.
Note that $\hat D_\mathrm{z}$ commutes with both $\hat C$ and $\hat D_\mathrm{x}$.
For short measurement times such that $N\gamma_\mathrm{x}^2 t^2\ll1$ and $N\omega^2 t^2\ll1$, by approximation of Trotter expansion, it can be shown that $\rho(t)\approx\exp(\gamma_\mathrm{z} t\hat D_\mathrm{z}/2)\exp(\gamma_\mathrm{x} t\hat D_\mathrm{x}/2)\exp(\omega t\hat C)\rho(0)$. 
Note that $\exp(\gamma_\mathrm{a} t\hat D_\mathrm{a}/2)=\prod_{i=1}^N\hat V_{\mathrm{a},i}$ for $\mathrm{a}=\mathrm{x,\,z}$ with $\hat V_{\mathrm{x},i}\colon\rho\mapsto(1-p_\mathrm{x})\rho+p_\mathrm{x}X_i\rho X_i$ and $\hat V_{\mathrm{z},i}\colon\rho\mapsto(1-p_\mathrm{z})\rho+p_\mathrm{z}Z_i\rho Z_i$, where $p_\mathrm{a}=(1-e^{-\gamma_\mathrm{a}t})/2$.
In the case of the raw GHZ state scenario (i.e., $n=1$), we obtain the exact result of the QFI about $\omega$ as 
\begin{multline}\label{eq:QFI}
    F[\rho(t)] = N^2 t^2 - N^2 t^2\sum_{k=0}^{(N-1)/2}\binom{N}{k} \\
    \frac{2a_k(1-x_k^2)(1-y_k^2)}
    {2-x_k^2-y_k^2+(y_k^2-x_k^2)\cos2N\omega t},
\end{multline}
where 
\begin{align}\label{eq:parameters}
a_k =& p_\mathrm{x}^k(1-p_\mathrm{x})^{N-k}+p_\mathrm{x}^{N-k}(1-p_\mathrm{x})^k, \nonumber\\
x_k =& (1-2p_\mathrm{z})^N,\nonumber\\
y_k =& (1-2p_\mathrm{z})^N \frac{(1-p_\mathrm{x})^{N-2k}-p_\mathrm{x}^{N-2k}}{(1-p_\mathrm{x})^{N-2k}+p_\mathrm{x}^{N-2k}}
\end{align}
(See Supplementary Note 2 for the detailed calculations).
When $\gamma_\mathrm{x}=0$, we have $F[\rho(t)]=e^{-2N\gamma_\mathrm{z} t}N^2t^2$, which is consistent with the result of ref.~\cite{Huelga1997}.
When $\gamma_\mathrm{z}=0$, we have $F[\rho(t)]=N^2t^2$, which is consistent with theorem~3 as the QFI is totally preserved if there are only bit-flip errors.
Note that this is not explicit if we only consider the protection of quantum states.
Moreover, as long as $p_\mathrm{z}$ is small such that $(1-2p_\mathrm{z})^N$ is close to $1$, the QFI of the noisy states is close to $N^2 t^2$ and insensitive to $p_\mathrm{x}$.
Therefore, the phase-flip code is enough to protect the QFI against the parallel and transversal dephasing in such a situation.

\begin{figure}[tb]
	\begin{center}
		\includegraphics[width=85mm]{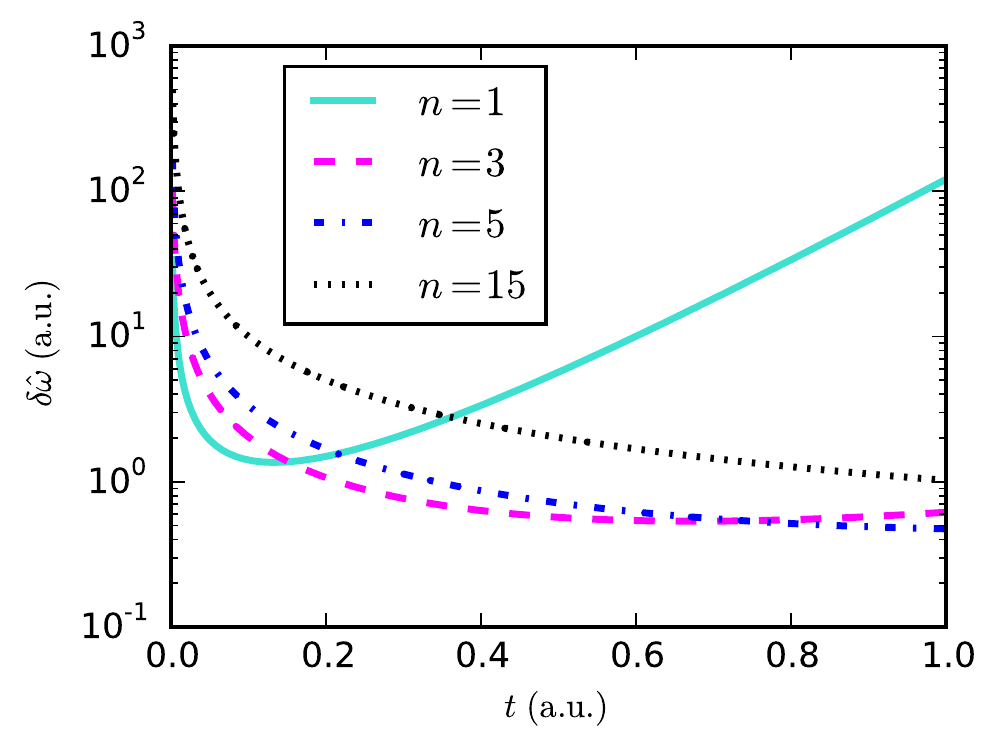}
	\end{center}
	\caption{\label{fig:QFI_time}
	{\bf Quantum Cram\'er-Rao bounds of a single shot measurement.} 
	The estimation error $\delta\hat{\omega}:=\langle(\hat\omega/|\frac{\mathrm{d}\braket{\hat\omega}}{\mathrm{d}\omega}|-\omega)^2\rangle^{1/2}$, represented by the $y$-axis, is bounded from below by the quantum Cram\'er-Rao bounds (the curves in the figure). 
	The $x$-axis represents the time of the signal accumulation process.
	Here, the total number of qubits is $N=15$, which is divided into blocks of size $n=1$, $3$, $5$, and $15$.
	Each block is protected by an $n$-qubit phase-flip code.
	The probe state is the logical GHZ state, and the sensing transformation is the independent logical phase accumulation.
	The figure is plotted at $\omega=0.001$, $\gamma_\mathrm{x}=0.001$, and  $\gamma_\mathrm{z}=0.5$.
}
\end{figure}

\begin{figure}[tb]
	\begin{center}
		\includegraphics[width=85mm]{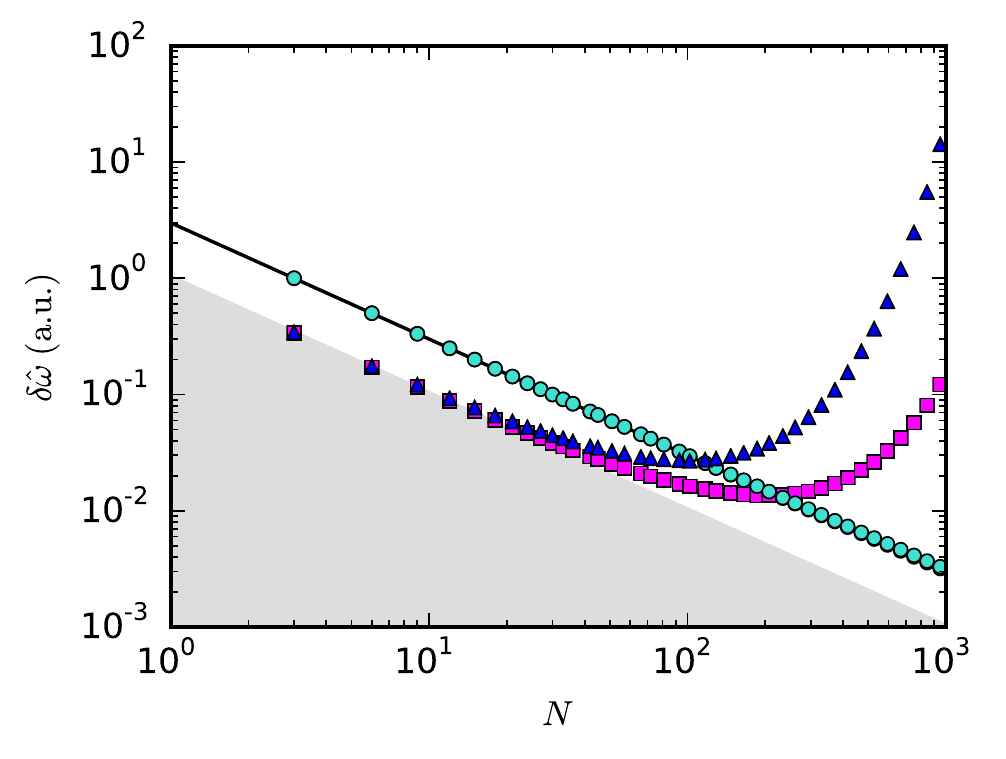}
	\end{center}
	\caption{\label{fig:scaling}
	{\bf Quantum Cram\'er-Rao bounds as a function of the total number of qubits.}
	The gray region denotes where the uncertainty is below the Heisenberg limit $1/N$. 
	In the scenario of using the raw GHZ states and the raw phase accumulation, the quantum Cram\'er-Rao bounds cannot follow the Heisenberg scale when $N$ is large, see the magenta square markers for the first kind of noise with $\gamma_\mathrm{x} = 5\times10^{-4}$ and $\gamma_\mathrm{z} = 5\times10^{-3}$, and the blue triangle markers for the second kind of  noise with $\gamma_\mathrm{x} = 10^{-3}$ and $\gamma_\mathrm{z} = 10^{-2}$. 
	In the scenario of using the logical GHZ states and the logical phase accumulation with the three-qubit phase-flip code in each block, the quantum Cram\'er-Rao bounds for the first and second kinds of noise are so close that they are visually indistinguishable in the figure (denoted by turquoise circle marker), which reflects the robustness of the scheme. 
	Furthermore, in the logical scenario, the quantum Cram\'er-Rao bounds follow well the Heisenberg scale of $3/N$ (black solid line).
	The figure is plotted at $\omega = 0.001$ and $t = 1$.
	}
\end{figure}

The logical GHZ state scenario, where $n$ is an odd number greater than one, can be recast into the raw GHZ state scenario by noting that only logical errors remain after the error correction in each block. 
Since less than or equal to $(n-1)/2$ phase-flip errors are corrected, the probability of the logical phase-flip error is suppressed to $\bar p_\mathrm{z}=\sum_{k=0}^{(n-1)/2}\binom{n}{k}p_\mathrm{z}^{n-k}(1-p_\mathrm{z})^k$, see ref.~\cite{Dur2014}.
Note that each single-qubit bit-flip error is equivalent to a logical bit-flip error on that block.
The probability of the logical bit-flip error is $\bar p_\mathrm{x}=[1-(1-2p_\mathrm{x})^n]/2$. Then, the QFI of the noisy states in the logical GHZ state scenario can also be given through equations~(\ref{eq:QFI}) and (\ref{eq:parameters}) via substituting $p_\mathrm{x}$, $p_\mathrm{z}$, and $N$ by $\bar p_\mathrm{x}$, $\bar p_\mathrm{z}$, and $\lfloor N/n \rfloor$, respectively.
Note that $\bar p_\mathrm{x}>p_\mathrm{x}$ for $0<p_\mathrm{x}<1/2$, which means that the bit-flip errors are amplified.
However, as long as the phase-flip errors are sufficiently suppressed, the bit-flip errors in the GHZ state scenario are almost harmless to QFI.

For instance, it is shown in Fig.~\ref{fig:QFI_time} that in the raw GHZ state scenario the quantum Cram\'er-Rao bound rises when the destructive effect of the noise on the QFI is dominant over the gain of the signal accumulation, whereas the use of the error correction can suppress the phase-flip errors so that the quantum states can gain more QFI by the signal accumulation process for a long time even in the presence of the bit-flip errors.
In Fig.~\ref{fig:scaling}, we show that using the phase-flip QEC code the quantum Cram\'er-Rao bound has a Heisenberg scaling with a constant factor, i.e., $\delta\hat\omega\sim n/N$. Note that here the size $n$ of the block is a fixed small integer, e.g., $n=3$ for the three-qubit phase-flip code in each block. Therefore, when $N$ becomes large, $n/N$ will be much smaller than the standard quantum limit $1/\sqrt{N}$.
In this comparison, the resources are measured by the total number of the physical qubits used. 
However, it should be noted that the signal accumulation processes are different in the two scenarios.

\vspace{8pt}
\noindent\textbf{Discussion}\\
The scenario considered in this work assumes a model where the noise occurs after the signal sensing. 
This coincides with analog communication over noisy quantum channels~\cite{Personick1971}, where analog signals are encoded in quantum states, transmitted over a noisy quantum channel, and estimated by the receivers.
The assumed model is also applicable for the sensing-stage noise whose generator is commuting with that of the signal sensing. 
Some known examples belonging to this class include the depolarization, the dephasing, and the spontaneous emission~\cite{Demkowicz-Dobrzanski2012} in the two-level systems with the signal sensing generated by the Pauli $Z$ matrix, and the photon loss~\cite{Dorner2009} and the phase diffusion~\cite{Genoni2011} in the optical fields with the signal sensing generated by the photon number operator.
For realistic instruments where the noise may be very complicated, our method can be applied together with other technologies such as dynamical decoupling~\cite{Arrad2014}.

In summary, we have established a theory of error correction designed for quantum metrology in the context of quantum estimation theory. 
The purpose of our specialized QEC is to preserve the QFI, which determines the best precision of estimating the value of a parameter, instead of the quantum states themselves.  
We have given testable conditions to identify the errors to which the QFI is immune, and constructed the optimal measurements in noisy states for the best estimation precision.
While in the standard QEC any states, mixed or pure, in the coding subspace can be used in a metrological scheme, our conditions do not generally give rise to a subspace, instead only a special set of states that can serve our purpose.
Our method can be readily applied for some parameter estimation problems, especially for those in the stabilizer formalism. 
Comparing with the standard stabilizer codes, our theory has the advantages of, firstly,  being capable of preserving QFI against more errors using the same amount of resources and, secondly, sparing the recovery operations.

\vspace{8pt}
\noindent\textbf{Methods}\\
\textbf{ Condition for preserving QFI}.
Let us start with a crucial observation on the loss of QFI after a known noisy channel. 
For a given channel $\cN$ with Kraus operators $\{E_j\}$, we denote  $V$ a unitary representation on the system plus an ancilla in the state $\rho_a$ such that 
\begin{equation}
	\cN(\rho_\theta)=\sum_jE_j\rho_\theta E_j^\dagger=\tr_a\big[V(\rho_\theta\otimes\rho_a)V^\dagger\big],
\end{equation}
where $\tr_a$ is the partial trace over the ancilla. 
For the sake of rigorousness, we assume that we always have bounded SLD operators henceforth, i.e., SLD operators $L_\theta$ and $\cL_\theta$ for states $\varrho_\theta$ and $\cN(\rho_\theta)$, respectively, exist and are finite. 
The loss of QFI can be expressed as (see Supplementary Note 1)
\begin{equation}\label{eq:rms} 	
\begin{aligned}
\Delta_\mathrm{F}(\rho_\theta,\cN)&:=F(\rho_\theta)-F(\cN(\rho_\theta))\\
&=\min_{Q\in\mathrm{Herm}(\cN)}\epsilon(\rho_\theta,\cN,Q)\\
&=\sum_j\lVert(\cL_\theta E_j-E_j L_\theta)\sqrt{\rho_\theta}\rVert_\mathrm{HS}^2,
\end{aligned}
\end{equation}
where $\mathrm{Herm}(\cN)$ denotes the set of all bounded Hermitian operators on the Hilbert space associated with the output of $\cN$, $\lVert O\rVert_\mathrm{HS}:=\sqrt{\tr(O^\dagger O)}$ is the Hilbert-Schmidt norm of the operator $O$, and
\begin{equation}
    \epsilon(\rho_\theta,\cN,Q):=\langle\big(V^\dagger(Q\otimes\bm{1})V-L_\theta\otimes\bm{1}\big)^2\rangle_{\rho_\theta\otimes\rho_a}
\end{equation}
is exactly the square of the measurement error used by Ozawa to derive his error-disturbance uncertainty relation~\cite{Ozawa2003}. From equation (\ref{eq:rms}), we see that the loss of QFI can be understood as the minimal measurement error of measuring a Hermitian operator $Q$ after the given noisy channel compared with measuring $L_\theta$ before the noisy channel. 
We note that due to equation (\ref{eq:rms}), the following statements are equivalent:
\begin{enumerate}
\item[(a)] $\Delta_\mathrm{F}(\rho_\theta,\cN)=0$.
\item[(b)] There exists a Hermitian operator $Q$ such that $QE_j\sqrt{\rho_\theta}=E_jL_\theta\sqrt{\rho_\theta}$ is satisfied for all $j$.
\item[(c)] $\cL_\theta E_j\sqrt{\rho_\theta}=E_jL_\theta\sqrt{\rho_\theta}$ is satisfied for all $j$.
\end{enumerate}

\textbf{Sketch of the proof}.
The necessary and sufficient condition (\ref{eq:NS}) for the preservation of QFI under a known noisy channel follows from the equivalence between (a) and (c). 
Theorem~\ref{thm:discretization} is a consequence of equation~(\ref{eq:NS}).
The geometric picture illustrated in Fig.~\ref{fig:geo} is due to the equivalence between (a) and (b).
Theorem~\ref{thm:testable} is implied by the equivalence between (a) and (b) together with the following lemma, for which we give a constructive proof in Supplementary Note 1.

\begin{lemma}\label{lem:math}
For two indexed families of vectors $\ket{s_j}$ and $\ket{d_j}$, there exists a Hermitian operator $Q$ such that $Q \ket{s_j} = \ket{d_j}$ for all $j$, if and only if \textup{(i)} $\langle s_j|d_k\rangle=\langle d_j|s_k\rangle$ 
for all $j$ and $k$ and \textup{(ii)} for all $\alpha_j$ such that $\sum_j\alpha_j\ket{s_j}=0$, $\sum_j\alpha_j\ket{d_j}=0$ must be satisfied.
\end{lemma}

Theorem~\ref{thm:qec} follows from the satisfaction of the two testable conditions in Theorem~\ref{thm:testable} for the errors $\{E_j\bar X^\tau\}$ with $\tau=0,1$, where $\bar X$ is a Pauli error that commutes with the stabilizer of the code $S$ and anticommutes with $\Delta H_\mathrm{eff}$. 
The full proof is presented in Supplementary Note 1.

The theorems and the lemma in this paper are also valid for infinite dimensional systems, as long as the SLD operator for the given parametric family of states is bounded.

\vspace{8pt}
{\noindent\bf Acknowledgments\\}
Discussions with Mankei Tsang, Ranjith Nair, and Pei-Qing Jin are gratefully acknowledged.
This work is funded by the Singapore Ministry of Education (partly through the Academic Research Fund Tier 3 MOE2012-T3-1-009), the National Research Foundation, Singapore (Grant No.\ WBS: R-710-000-008-271 and Grant No.\ NRF-NRFF2011-07), and the National Natural Science Foundation of China (No. 11304196).

\def\bibsection{\section*{Supplementary References}}

\renewcommand{\theequation}{S\arabic{equation}}
\setcounter{equation}{0}
\widetext

\section*{Supplementary Note 1: Detailed derivations of some results in the main text.\label{sec:theorems}}
\begin{proof}[Proof of equation (\ref{eq:rms}) given in the main text]
Let $Q$ be an arbitrary Hermitian operator on the Hilbert space associated with the outputs of a noisy channel $\cN$, and $\cL_\theta$ be an SLD operator for the noisy states $\cN(\rho_\theta)$ and assumed bounded. Then, 
\begin{align}
    \epsilon(\rho_\theta,\cN,Q) &= \big\langle\big(L_\theta\otimes\bm{1}-V^\dagger(Q\otimes\bm{1})V\big)^2\big\rangle_{\rho_\theta\otimes\rho_a}\nonumber\\
    &= F(\rho_\theta)+\tr\big[Q^2\cN(\rho_\theta)\big]-\tr\big[Q\cN\big(\{L_\theta,\rho_\theta\}\big)\big]\nonumber\\
    &= F(\rho_\theta)+\tr\big[Q^2\cN(\rho_\theta)\big]-\tr\big[Q\{\cL_\theta,\cN(\rho_\theta)\}\big]\nonumber\\
    &= \Delta_\mathrm{F}(\rho_\theta,\cN)+\tr\big[(Q-\cL_\theta)^2\cN(\rho_\theta)\big] \label{eq:s1},
\end{align}
where we have used the following relations:
\begin{equation}
\cN\left(\{L_\theta,\rho_\theta\}\right) 
=2\,\cN\left(\frac{\mathrm{d}\rho_\theta}{\mathrm{d}\theta}\right)
=2\,\frac{\mathrm{d}\cN(\rho_\theta)}{\mathrm{d}\theta}
=\big\{\cL_\theta,\cN(\rho_\theta)\big\}.
\end{equation}
As a result of equation (\ref{eq:s1}), we have $\Delta_\mathrm{F}(\rho_\theta,\cN)\leq\epsilon(\rho_\theta,\cN,Q)$ for every Hermitian operator $Q$, and the equality holds if (and only if) $Q\cN(\rho_\theta)=\cL_\theta\cN(\rho_\theta)$. Thus, we obtain $\Delta_\mathrm{F}(\rho_\theta,\cN)=\min_{Q\in\mathrm{Herm}(\cN)}\epsilon(\rho_\theta,\cN,Q)$. 
Moreover, by using the Kraus operators $E_j$ for $\cN$, it can be readily checked that 
\begin{equation}\label{eq:Kraus}
\epsilon(\rho_\theta,\cN,Q)= \sum_j\lVert(Q E_j-E_j L_\theta)\sqrt{\rho_\theta}\rVert_\mathrm{HS}^2
\end{equation}
with $\lVert O\rVert_\mathrm{HS}:=\sqrt{\tr(O^\dagger O)}$ being the Hilbert-Schmidt norm of the operator $O$. Thus we obtain $\Delta_\mathrm{F}(\rho_\theta,\cN)=\sum_j\lVert(\cL_\theta E_j-E_j L_\theta)\sqrt{\rho_\theta}\rVert_\mathrm{HS}^2$ by taking $Q=\cL_\theta$.
\end{proof}

\begin{proof}[Proof of Theorem~\ref{thm:testable} given in the main text] Consider the case of pure state $|\psi_\theta\rangle$ first.
As a result of the necessary and sufficient condition, i.e., equation (\ref{eq:NS}) in the main text, and Theorem~\ref{thm:discretization} in the main text, QFI of $\ket{\psi_\theta}$ is preserved against a set $\{E_j\}$ of errors if and only if $\cL_\theta E_j\ket{\psi_\theta} = E_j L_\theta\ket{\psi_\theta}$ are satisfied for all $j$. 
Then, conditions (i) and (ii) in Theorem~\ref{thm:testable} in the main text are results of Lemma~\ref{lem:math} in the main text for two indexed families $\{E_j\ket{\psi_\theta}\}$ and $\{2 E_j \ket{\mathscr{D}_\theta\psi_\theta}\}$.

In the case of mixed states $\rho_\theta$, equation (\ref{eq:NS}) in the main text is equivalent to $\cL_\theta E_j|\psi_\theta\rangle=E_jL_\theta|\psi_\theta\rangle$ for all states $|\psi_\theta\rangle$ in the range of $\rho_\theta$. Choosing a set of linearly independent states $\{|\psi_{l,\theta}\rangle\}$ in the range of $\rho_\theta$, e.g., the eigenstates of $\rho_\theta$ corresponding to nonzero eigenvalues, applying Lemma~\ref{lem:math} in the main text to the two indexed families of states $\{E_j|\psi_{l,\theta}\rangle\}$ and $\{E_j L_\theta|\psi_{l,\theta}\rangle\}$ with composite index $(j,l)$, we obtain the testable conditions for mixed states.
\end{proof}

\begin{proof}[Proof of Lemma~\ref{lem:math} given in the main text]
{\it Necessity.}
If there exists a Hermitian operator $Q$ such that $Q\ket{s_j}=\ket{d_j}$ for all $j$, then condition (i) can be obtained by using the hermicity of $Q$ as 
\begin{equation}
    \langle s_j|d_k\rangle=\langle s_j|Q|s_k\rangle=\langle d_j|s_k\rangle,
\end{equation}
whilst, condition (ii) is obvious by noting that $\sum_j\alpha_j\ket{d_j}=Q\sum_j\alpha_j\ket{s_j}$, which vanishes if $\sum_j\alpha_j\ket{s_j}$ does.

{\it Sufficiency.}
Assume that conditions (i) and (ii) are satisfied, then we shall explicitly construct a desired Hermitian $Q$ such that $Q\ket{s_j}=\ket{d_j}$ for all $j$. 
Firstly, let us choose a maximal subset of linearly independent vectors $\ket{s_j}$ and denote the set of corresponding indices by $\mathbb{J}$.
Secondly, condition (ii) implies that restricting $j$ in $\mathbb{J}$ is sufficient for constructing the $Q$.
Note that every $\ket{s_{j^\prime}}$ with $j^\prime\notin \mathbb{J}$ can be expressed as 
\begin{equation}\label{eq:linearly}
    \ket{s_{j^\prime}}=\sum_{j\in\mathbb{J}}\alpha_j\ket{s_j}
\end{equation}
with $\alpha_j$ being complex numbers.
If $Q\ket{s_j}=\ket{d_j}$ for all $j\in\mathbb{J}$, then from condition (ii) and equation \eqref{eq:linearly}, we have
\begin{equation}
    \ket{d_{j^\prime}}=\sum_{j\in\mathbb{J}}\alpha_j\ket{d_j} 
   =Q\sum_{j\in\mathbb{J}}\alpha_j\ket{s_j} 
   =Q\ket{s_{j^\prime}}.
\end{equation}
Therefore, condition (ii) ensures that every Hermitian operator $Q$ satisfying $Q\ket{s_j} = \ket{d_j}$ for all $j$ restricted in $\mathbb{J}$ must satisfy that for all $j$.
Thirdly, we explicitly construct $Q$ as follows. 
Condition (i) implies that the matrix $g$ defined by $g_{jk} = \langle s_j|d_k \rangle$ is Hermitian.
Therefore, $g$ can be diagonalized as 
\begin{equation}
    (u^\dagger gu)_{jk}=\braket{\tilde s_j|\tilde d_k}=c_j\delta_{jk},
\end{equation}
where $u$ is a unitary matrix, 
\begin{equation}
    \ket{\tilde s_j}=\sum_k u_{jk}\ket{s_k}, \qquad 
    \ket{\tilde d_j}=\sum_k u_{jk}\ket{d_k}, 
\end{equation}
and $c_j$ are real numbers.
Then, the Hermitian operator $Q$ is explicitly constructed as
\begin{equation}\label{eq:Q_opt}
    Q = \sum_{j|c_j\neq 0}\frac{1}{c_j}\ket{\tilde d_j}\bra{\tilde d_j} 
    + \sum_{j|c_j=0}\ket{\tilde d_j}\bra{\tilde s_j^\perp}+\ket{\tilde s_j^\perp}\bra{\tilde d_j},
\end{equation}
where $\ket{\tilde s_j^\perp}$ are vectors satisfying $\braket{\tilde s_j^\perp|\tilde s_k}=\delta_{jk}$ for all $j$ and $k$.
Such a set of $\ket{\tilde s_j^\perp}$ always exists as long as $\ket{\tilde s_j}$ are linearly independent~\cite{Chefles1998}.
It is easy to check that the Hermitian operator $Q$ defined above satisfies $ Q\ket{\tilde s_j} = \ket{\tilde d_j}$ for all $j$.
Therefore, 
\begin{align}
    Q\ket{s_j}=Q \sum_k u_{kj}^*\ket{\tilde s_k}=\sum_k u_{kj}^*\ket{\tilde d_k}\nonumber \\
    =\sum_l\Big(\sum_k u_{kj}^*u_{kl}\Big)\ket{d_l}=\ket{d_j}
\end{align}
are satisfied for all $j$ and the proof is completed.
\end{proof}

\begin{proof}[Proof of Theorem~\ref{thm:qec} given in the main text] 
We have only to check that the two conditions in Theorem 2 in the main text are satisfied for the errors $\{E_j\bar X^\tau\}$ with $\tau=0,1$, where $\bar X$ is a Pauli error that commutes with the stabilizers of the code and anticommutes with $\Delta H_\mathrm{eff}$. 
The first condition reads
\begin{equation}\label{eq:first}
\langle\psi_\theta|L_\theta\bar X^\gamma E^\dagger_k E_j\bar X^\tau |\psi_\theta\rangle=\langle\psi_\theta|\bar X^\gamma E^\dagger_k E_j\bar X^\tau L_\theta|\psi_\theta\rangle,
\end{equation}
which holds true due to the following facts: (i) When $\gamma=\tau$, or $\gamma\not=\tau$ and $E_k^\dagger E_j$ is detectable, then both sides of equation (\ref{eq:first}) vanish. (ii) When $\gamma\not=\tau$ and $E_k^\dagger E_j$ is a stabilizer, equation~(\ref{eq:first}) is ensured by the facts that $L_\theta|\psi_\theta\rangle=-2i\Delta H_\mathrm{eff}|\psi_\theta\rangle$ and  $\bar X$ anticommutes with $\Delta H_\mathrm{eff}$.  

To show that the second condition is also satisfied, we at first identify an independent set $\{E_j\}_{j\in J}$ of errors such that $E_k^\dagger E_j$ are detectable errors for arbitrary $j,k\in J$.
We denote by $J_j$ the set of indices $l$ such that $E_j^\dagger E_l$ is a stabilizer. 
It is easy to check that $\{E_j(1\pm\bar X)|\psi_\theta\rangle\}_{j\in J}$, as well as $\{E_j(1\pm\bar X)L_\theta|\psi_\theta\rangle\}_{j\in J}$, is a set of mutually orthogonal states. 
As a result, both sets $\{E_j \bar X^\tau |\psi_\theta\rangle\}_{j\in J}$ and $\{E_j \bar X^\tau L_\theta|\psi_\theta\rangle\}_{j\in J}$ are linearly independent. 
If there are complex numbers $\alpha_{j\tau}$ such that $\sum_{j,\tau}\alpha_{j\tau} E_j\bar X^\tau|\psi_\theta\rangle=0$, then we have $\sum_{j\in J,\tau}\alpha_{j\tau}^\prime E_j\bar X^\tau|\psi_\theta\rangle=0$ with  $\alpha_{j,\tau}^\prime=\sum_{l\in J_j} \alpha_{l,\tau}$, from which it follows that  $\alpha^\prime_{j\tau}=0$ for all $j\in J$ and $\tau=0,1$ and 
\begin{equation}
    \sum_{j,\tau}\alpha_{j\tau} E_j\bar X^\tau L_\theta|\psi_\theta\rangle
    =\sum_{j\in J,\tau}\alpha_{j\tau}^\prime E_j\bar X^\tau L_\theta|\psi_\theta\rangle=0.
\end{equation}

In order to investigate the property of the optimal measurement, we denote
\begin{equation}
    Q=\sum_{j\in J}E_j\frac{\{\bar X,L_\theta |\psi_\theta\rangle\langle\psi_\theta|L_\theta\}}{\langle \psi_\theta|L_\theta\bar X|\psi_\theta\rangle}E_j^\dagger.    
\end{equation}
Due to the fact that the errors $E_j$ with the index $j$ being restricted in $J$ are independent, it is easy to check that 
\begin{equation}
\bar X^\tau E_j^\dagger QE_j\bar X^\tau|\psi_\theta\rangle=Q|\psi_\theta\rangle=L_\theta|\psi_\theta\rangle
\end{equation}
for arbitrary $\tau=0,1$ and $j$. 
It follows that $QE_j\bar X^\tau|\psi_\theta\rangle=E_j\bar X^\tau L_\theta|\psi_\theta\rangle$, therefore this Hermitian operator $Q$ is an SLD operator not only for $|\psi_\theta\rangle$ but also for noisy states under the set $\{E_j\}$ of errors, therefore, the measurement with respect to the eigenstates of $Q$ is optimal for the noisy states.
Moreover, $Q$ commutes with all the stabilizers and $\bar X$, so they have common eigenstates. 
When the code space is $2$-dimensional, all the stabilizer generator and $\bar X$ constitute a complete set of mutually commuting observables, therefore  the joint measurement of them is equivalent to that with respect to eigenstates of $Q$.
\end{proof}

\section*{Supplementary Note 2: Calculations for the example\label{sec:example}}

Let us consider the physical system that implements the metrological scheme in the main text. The whole system is composed of $N$ (physical) qubits. The dynamical equation is given by
\begin{equation}
    \frac{\mathrm{d}\rho(t)}{\mathrm{d}t} = \omega\hat C\rho(t)+ 
    \frac{\gamma_\mathrm{x}}2\hat D_\mathrm{x}\rho(t)+
    \frac{\gamma_\mathrm{y}}2\hat D_\mathrm{y}\rho(t)+
    \frac{\gamma_\mathrm{z}}2\hat D_\mathrm{z}\rho(t)
\end{equation}
where $\omega$ is the signal parameter to be estimated, $\gamma_\alpha$ for $\alpha=\mathrm{x,y,z}$ are the strengths of the noises. The coherent evolution sensing the signal parameter $\omega$ is described by the superoperator $\hat C\colon\rho\to-i[H,\rho]$ with $H$ being a Hermitian operator. The uncorrelated dephasing noises along different directions are described by the superoperators 
$\hat D_\mathrm{x}\colon\rho\to\sum_{i=1}^N (X_i\rho X_i-\rho)$, 
$\hat D_\mathrm{y}\colon\rho\to\sum_{i=1}^N (Y_i\rho Y_i-\rho)$, and 
$\hat D_\mathrm{z}\colon\rho\to\sum_{i=1}^N (Z_i\rho Z_i-\rho)$.
Henceforth, we assume that $\gamma_\mathrm{y}=0$.

We consider two scenarios with different generating operator $H$ of the sensing evolution and different probe states, i.e., initial states $\rho(0)$. 
(i)
In the raw GHZ-state scenario, $H=\frac12\sum_{i} Z_i$ and $\rho(0)=|\psi_{\rm rGHZ}\rangle\langle\psi_{\rm rGHZ}|$ with $|\psi_{\rm rGHZ}\rangle:=\frac1{\sqrt2}(|0\rangle^{\otimes N}+|1\rangle^{\otimes N})$ being the $N$-qubit GHZ state. 
(ii)
In the logical GHZ-state scenario, we divide $N$ qubits into $m$ blocks, each of which is composed of $n$ qubits. 
Here, $N=nm$ is assumed.
Let $O^{[i]}_j$ be the operator $O$ of the $j$th qubit in the $i$th block. 
The generating operator of the coherent evolution is $H=\frac12\sum_i\bar Z^{[i]}$ with $\bar Z^{[i]}:=\prod_jZ^{[i]}_j$ being the logical $Z$ operator for the $i$th block, while the initial state is $\rho(0)=|\psi_{\rm lGHZ}\rangle\langle\psi_{\rm lGHZ}|$ with $|\psi_{\rm lGHZ}\rangle:=\frac1{\sqrt2}(|\bar0\rangle^{\otimes m}+|\bar1\rangle^{\otimes m})$ being the logical GHZ state, where $|\bar0\rangle=\frac1{\sqrt2}(|+\rangle^{\otimes n}+|-\rangle^{\otimes n})$ and $|\bar1\rangle=\frac1{\sqrt2}(|+\rangle^{\otimes n}-|-\rangle^{\otimes n})$ are logical basis for the $n$-qubit phase-flip code in each block. 

In both these two scenarios, the superoperator $\hat D_\mathrm{z}$ commutes with $\hat C$, $\hat D_\mathrm{x}$, respectively, therefore the evolution of the total system is given by $\exp\left(\frac{\gamma_\mathrm{z} t}{2}\hat D_\mathrm{z}\right)\exp\left(\omega t\hat C+\frac{\gamma_\mathrm{x} t}{2}\hat D_\mathrm{x}\right)$.
As long as $N\gamma_\mathrm{x}^2t^2\ll1$ and $N\omega^2t^2\ll1$, according to the Trotter expansion, the evolution is well approximated by $\rho(t)\approx\cN(\rho_\omega)$ with 
\begin{equation}
\cN=\exp\left(\frac{\gamma_\mathrm{z} t}{2}\hat D_\mathrm{z}\right) \exp\left(\frac{\gamma_\mathrm{x} t}{2}\hat D_\mathrm{x}\right)
\quad\mbox{and}\quad
\rho_\omega=\exp\left(\omega t\hat C\right)\rho(0).    
\end{equation}
It can be shown by some algebras that 
\begin{equation}
    \exp\left(\frac{\gamma_\mathrm{x} t}{2}\hat D_\mathrm{x}\right) = \prod_{i=1}^N\hat V_{\mathrm{x},i}
    \quad\mbox{and}\quad
    \exp\left(\frac{\gamma_\mathrm{z} t}{2}\hat D_\mathrm{z}\right) = \prod_{i=1}^N\hat V_{\mathrm{z},i},
\end{equation}
where the superoperators $\hat V_{\mathrm{x},i}$ and $\hat V_{\mathrm{z},i}$ are defined by $\hat V_{\mathrm{x},i}:\rho \mapsto (1-p_\mathrm{x})\rho+p_\mathrm{x}X_i\rho X_i$ and $\hat V_{\mathrm{z},i}:\rho \mapsto (1-p_\mathrm{z})\rho+p_\mathrm{z} Z_i\rho Z_i$ with $p_\mathrm{x}=(1-e^{-\gamma_\mathrm{x} t})/2$ and $p_\mathrm{z}=(1-e^{-\gamma_\mathrm{z} t})/2$.

In the raw GHZ-state scenario, the parametric states read $\rho_\omega=|\psi_\omega\rangle\langle\psi_\omega|$ with 
\begin{equation}
    |\psi_\omega\rangle:=\exp\Big(-\frac{i\omega t}{2} \sum_{i=1}^N Z_i\Big)|\psi_\mathrm{rGHZ}\rangle=\frac1{\sqrt2}(e^{-iN\omega t/2}|0\rangle^{\otimes N}+e^{iN\omega t/2}|1\rangle^{\otimes N}).
\end{equation}
Let us denote by $\hat O$ the superoperator defined by $\hat O\rho=O\rho O^\dagger$ for an operator $O$.
Since $\hat X_i \hat Z_j = \hat Z_j \hat X_i$ for all $i$ and $j$, the effective noisy channel $\cN$ can be expressed by 
\begin{equation}
    \cN = \prod_{i=1}^N (1-p_\mathrm{x}+p_\mathrm{x} \hat X_i) 
                        (1-p_\mathrm{z}+p_\mathrm{z} \hat Z_i).
\end{equation}
Note that $\rho_\omega$ is in the subspace $\cC$ spanned by $|0\rangle^{\otimes N}$ and $|1\rangle^{\otimes N}$, which is a bit-flip code capable of correcting less than or equal to $(N-1)/2$ bit-flip ($X$) errors.
The logical Pauli operators can be defined by $X_\mathrm{L}=\prod_{i=1}^N X_i$, $Z_\mathrm{L}=\prod_{i=1}^N Z_i$, and $Y_\mathrm{L} = -i Z_\mathrm{L} X_\mathrm{L}$, respectively. 
On this bit-flip code space $\cC$, each $\hat Z_i$ has the same effect as $\hat Z_\mathrm{L}$, therefore 
\begin{align}
    \cN|_\cC &= (1-p_\mathrm{z}+p_\mathrm{z} \hat Z_\mathrm{L})^N \prod_{i=1}^N (1-p_\mathrm{x}+p_\mathrm{x} \hat X_i) \\
    &= \left[\frac{1+(1-2p_\mathrm{z})^N}{2}+\frac{1-(1-2p_\mathrm{z})^N}{2}\hat Z_\mathrm{L}\right]
    \sum_{I_\mathrm{x} \subseteq I} (1-p_\mathrm{x})^{N-|I_\mathrm{x}|}\prod_{i\in I_\mathrm{x}}\hat X_i,
\end{align}
where $I$ denotes the index set for the qubits.
Moreover, different combinations of not larger than $(N-1)/2$ bit-flip errors map the code space $\cC$ into orthogonal subspaces.
Larger than $(N-1)/2$ bit-flip errors can be thought of as bit-flip errors occurring on the complement together with a logical $X$ error, namely, $\prod_{i\in I_\mathrm{x}}\hat X_i=\hat X_\mathrm{L}\prod_{i\notin I_\mathrm{x}}\hat X_i$ for an arbitrary index subset $I_\mathrm{x}$ for the qubits.
This indicates the following decomposition:
\begin{equation}\label{seq:NC}
    \cN|_\cC = \sum_{I_\mathrm{x}:|I_\mathrm{x}|\leq(N-1)/2}\hat S(|I_\mathrm{x}|)\prod_{i\in I_\mathrm{x}}\hat X_i,
\end{equation}
with
\begin{equation}
    \hat S(k) := [(1-p_\mathrm{x})^{N-k}+(1-p_\mathrm{x})^k\hat X_\mathrm{L}] \left[\frac{1+(1-2p_\mathrm{z})^N}{2}+\frac{1-(1-2p_\mathrm{z})^N}{2}\hat Z_\mathrm{L}\right].
\end{equation}
Different terms in the sum of equation (\ref{seq:NC}) map $\rho_\omega$, which is in the code space $\cC$, into orthogonal subspaces.
Because (i) every $\hat X_i$ is a unitary channel which preserves the QFI, and (ii) $F[\lambda\rho_\omega+(1-\lambda)\eta_\omega]=\lambda F(\rho_\omega)+(1-\lambda) F(\eta_\omega)$ when the parametric families $\rho_\omega$ and $\eta_\omega$ of density operators are in orthogonal subspaces, where $0\leq\lambda\leq1$, we get 
\begin{align}
    F[\cN(\rho_\omega)] &= \sum_{k=1}^{(N-1)/2}\binom{N}{k}a_k F\left[\frac{\hat S(k)\rho_\omega}{a_k}\right],
\end{align}
where $a_k$ given by
\begin{equation}
    a_k :=\tr[\hat S(k)\rho_\omega] = p_\mathrm{x}^k(1-p_\mathrm{x})^{N-k}+p_\mathrm{x}^{N-k}(1-p_\mathrm{x})^k.
\end{equation}
are independent of $\omega$. 
Each $\hat S(k)\rho_\omega/a_k$ is a $2\times2$ normalized density operator on $\cC$, therefore its QFI can be obtained by the explicit expression for $2\times2$ density matrices $\varrho_\omega$~\cite{Dittmann1999}:
\begin{equation}\label{seq:QFI2d}
    F(\varrho_\omega) = 
    \tr\left[ \frac{\partial\varrho_\omega}{\partial\omega}\frac{\partial\varrho_\omega}{\partial\omega} + \frac{1}{\det(\varrho_\omega)}(\mathbf{1}-\varrho_\omega)\frac{\partial\varrho_\omega}{\partial\omega} (\mathbf{1}-\varrho_\omega)\frac{\partial\varrho_\omega}{\partial\omega} \right].
\end{equation}
Reminding that 
$\rho_\omega=\frac12 P[1+\cos(N\omega t)\bar X+\sin(N\omega t)\bar Y]P$ with $P=|0\rangle\langle0|^{\otimes N}+|1\rangle\langle1|^{\otimes N}$, $\hat S(k)\rho_\omega/a_k$ can be expressed in the form of 
\begin{equation}
    \frac{\hat S(k)\rho_\omega}{a_k} = \frac12 P[1+x_k\cos(N \omega t)\bar X+y_k\sin(N \omega t)\bar Y]P,
\end{equation}
where the coefficients are given by
\begin{align}
    x_k =& (1-2p_\mathrm{z})^N,\nonumber\\
    y_k =& (1-2p_\mathrm{z})^N \frac{(1-p_\mathrm{x})^{N-2k}-p_\mathrm{x}^{N-2k}}{(1-p_\mathrm{x})^{N-2k}+p_\mathrm{x}^{N-2k}}
\end{align}
Through equation (\ref{seq:QFI2d}), we get the QFI of the normalized version of $\hat S(k)\rho_\omega$ as
\begin{equation}\label{seq:Fk}
    F_k\equiv F\left(\frac{\hat S(k)\rho_\omega}{a_k}\right)=\left[1-
    \frac{2(1-x_k^2)(1-y_k^2)}
    {2-x_k^2-y_k^2+(y_k^2-x_k^2)\cos2N\omega t}\right].
\end{equation}
Then, the QFI of the noisy states is given by $F[\rho(t)]=\sum_{k=0}^{(N-1)/2}\binom{N}{k}a_k F_k$.



\begin{thebibliography}{00}
\bibitem{Caves1981} Caves, C. M. 
Quantum-mechanical noise in an interferometer.
\textit{Phys. Rev. D} \textbf{23}, 1693--1708 (1981). 

\bibitem{Yurke1986} Yurke, B., McCall,  S. L. \& Klauder, J. R.
SU(2) and SU(1,1) interferometers.
\textit{Phys. Rev. A} \textbf{33}, 4033--4054 (1986). 

\bibitem{Wineland1992} Wineland, D. J., Bollinger, J. J., Itano, W. M., Moore, F. L. \& Heinzen, D. J.
Spin squeezing and reduced quantum noise in spectroscopy.
\textit{Phys. Rev. A} \textbf{46}, R6797--R6800 (1992).

\bibitem{Holland1993} Holland, M. J. \& Burnett K.
Interferometric detection of optical phase shifts at the Heisenberg limit.
\textit{Phys. Rev. Lett.} \textbf{71}, 1355--1358 (1993).

\bibitem{Dowling1998} Dowling, J. P.
Correlated input-port, matter-wave interferometer: Quantum-noise limits to the atom-laser gyroscope.
\textit{Phys. Rev. A} \textbf{57}, 4736--4746 (1998). 

\bibitem{Giovannetti2004} Giovannetti, V., Lloyd, S. \& Maccone., L.
Quantum-enhanced measurements: beating the standard quantum limit.
\textit{Science} \textbf{306}, 1330--1336 (2004).

\bibitem{Giovannetti2006} Giovannettki, V., Lloyd, S. \& Maccone, L. 
Quantum metrology.
\textit{Phys. Rev. Lett.} \textbf{96}, 010401 (2006).

\bibitem{Giovannetti2011} Giovannetti., V. Lloyd, S. \& Maccone., L.
Advances in quantum metrology.
\textit{Nature Photo.} \textbf{5}, 222--229 (2011).

\bibitem{Huelga1997} Huelga, S. F. \textit{et al.} 
Improvement of frequency standards with quantum entanglement.
\textit{Phys. Rev. Lett.} \textbf{79}, 3865--3868 (1997). 

\bibitem{Escher2011} Escher, B. M., de Matos Filho, R. L. \& Davidovich, L.
General framework for estimating the ultimate precision limit in noisy quantum-enhanced metrology.
\textit{Nat. Phys.} \textbf{7}, 406--411 (2011).

\bibitem{Demkowicz-Dobrzanski2012} Demkowicz-Dobrza\'nski, R., Ko\l{}ody\'nski, J. \& Gu\c{t}\u{a}, M.
The elusive Heisenberg limit in quantum-enhanced metrology.
\textit{Nat. Commun.} \textbf{3}, 1063 (2012).

\bibitem{Chaves2013} Chaves, R., Brask, J. B., Markiewicz, M., Ko\l{}ody\'{n}ski, J. \& Ac\'\i{}n, A.
Noisy metrology beyond the standard quantum limit.
\textit{Phys. Rev. Lett.} \textbf{111}, 120401 (2013).

\bibitem{Tsang2013} Tsang, M.
Quantum metrology with open dynamical systems.
\textit{New J. Phys.} \textbf{15}, 073005 (2013). 

\bibitem{Rubin2007} Rubin, M. A. \& Kaushik, S. 
Loss-induced limits to phase measurement precision with maximally entangled states.
\textit{Phys. Rev. A} \textbf{75}, 053805 (2007).

\bibitem{Huver2008} Huver, S. D., Wildfeuer, C. F. \& Dowling, J. P.
Entangled Fock states for robust quantum optical metrology, imaging, and sensing.
\textit{Phys. Rev. A} \textbf{78}, 063828 (2008).

\bibitem{Dorner2009} Dorner, U. \textit{et al.}
Optimal quantum phase estimation.
\textit{Phys. Rev. Lett.} \textbf{102}, 040403 (2009).

\bibitem{Demkowicz-Dobrzanski2009} Demkowicz-Dobrzanski, R. \textit{et al.} 
Quantum phase estimation with lossy interferometers.
\textit{Phys. Rev. A} \textbf{80}, 013825 (2009).

\bibitem{Lee2009} Lee, T.-W. \textit{et al.} 
Optimization of quantum interferometric metrological sensors in the presence of photon loss.
\textit{Phys. Rev. A} \textbf{80}, 063803 (2009).

\bibitem{Maccone2009} Maccone, L. \& De Cillis, G.  
Robust strategies for lossy quantum interferometry.
\textit{Phys. Rev. A} \textbf{79}, 023812 (2009).

\bibitem{Ono2010} Ono, T. \& Hofmann, H. F. 
Effects of photon losses on phase estimation near the Heisenberg limit using coherent light and squeezed vacuum. 
\textit{Phys. Rev. A} \textbf{81}, 033819 (2010). 

\bibitem{Joo2011} Joo, J., Munro,  W. J. \& Spiller, T. P.
Quantum metrology with entangled coherent states.
\textit{Phys. Rev. Lett.} \textbf{107}, 083601 (2011).

\bibitem{Jiang2012} Jiang, K. \textit{et al.}
Strategies for choosing path-entangled number states for optimal robust quantum-optical metrology in the presence of loss.
\textit{Phys. Rev. A} \textbf{86}, 013826 (2012).

\bibitem{Spagnolo2012} Spagnolo, N. \textit{et al.} 
Phase estimation via quantum interferometry for noisy detectors.
\textit{Phys. Rev. Lett.} \textbf{108}, 233602 (2012).

\bibitem{Kacprowicz2010} Kacprowicz, M., Demkowicz-Dobrzanski, R., Wasilewski, W., Banaszek, K. \& Walmsley, I. A. 
Experimental quantum-enhanced estimation of a lossy phase shift.
\textit{Nature Photon.} \textbf{4}, 357 (2010).

\bibitem{Genoni2011} Genoni, M. G., Olivares, S. \& Paris, M. G. A.
Optical phase estimation in the presence of phase diffusion.
\textit{Phys. Rev. Lett.} \textbf{106}, 153603 (2011).

\bibitem{Genoni2012} Genoni, M. G. \textit{et al.} 
Optical interferometry in the presence of large phase diffusion.
\textit{Phys. Rev. A} \textbf{85}, 043817 (2012).


\bibitem{Macchiavello2002} Macchiavello, C., Huelga, S.F., Cirac, J.I., Ekert, A.K. 
\& Plenio, M.B.
Decoherence and quantum error correction in frequency standards, 
in \textit{Quantum Communication, Computing, and Measurement 2}, p.\ 337--345
(Kluwer Academic Publishers, 2002).

\bibitem{Preskill2000} Preskill, J.
Quantum clock synchronization and quantum error correction.
Preprint at http://arxiv.org/abs/quant-ph/0010098 (2000).


\bibitem{Dur2014} D\"ur, W., Skotiniotis, M., Fr\"owis, F. \&  Kraus B. 
Improved quantum metrology using quantum error correction.
\textit{Phys. Rev. Lett.} \textbf{112}, 080801 (2014).

\bibitem{Kessler2014} Kessler, E. M., Lovchinsky, I., Sushkov, A. O. \& Lukin, M. D.
Quantum error correction for metrology.
\textit{Phys. Rev. Lett.} \textbf{112}, 150802 (2014).

\bibitem{Arrad2014} Arrad, G., Vinkler, Y., Aharonov, D. \& Retzker, A.
Increasing sensing resolution with error correction.
\textit{Phys. Rev. Lett.} \textbf{112}, 150801 (2014).

\bibitem{Ozeri2013} Ozeri, R.
Heisenberg limited metrology using quantum error-correction codes.
Preprint at http://arxiv.org/abs/1310.3432 (2013).

\bibitem{Shor1995} Shor, P. W.
Scheme for reducing decoherence in quantum computer memory.
\textit{Phys. Rev. A} \textbf{52}, R2493--R2496 (1995).

\bibitem{Bennett1996} Bennett, C. H., DiVincenzo, D. P., Smolin, J. A. \& Wootters, W. K.
Mixed-state entanglement and quantum error correction.
\textit{Phys. Rev. A} \textbf{54}, 3824--3851 (1996).

\bibitem{Steane1996} Steane, A. M.
Error correcting codes in quantum theory.
\textit{Phys. Rev. Lett.} \textbf{77}, 793--797 (1996).

\bibitem{Gottesman1996} Gottesman, D.
Class of quantum error-correcting codes saturating the quantum Hamming bound.
\textit{Phys. Rev. A} \textbf{54}, 1862--1868 (1996).

\bibitem{Knill1997} Knill, E. \& Laflamme, R.
Theory of quantum error-correcting codes.
\textit{Phys. Rev. A} \textbf{55}, 900--911 (1997).

\bibitem{Yu2008} Yu, S., Chen, Q., Lai, C. H., \& Oh, C. H.,
Nonadditive Quantum Error-Correcting Code.
\textit{Phys. Rev. Lett.} \textbf{101}, 090501 (2008).

\bibitem{Yu2013} Yu, S., Bierbrauer, J., Dong, Y., Chen, Q. \& Oh, C. H.
All the stabilizer codes of distance 3. 
\textit{{IEEE} Trans. Inform. Theory} \textbf{59}, 5179--5185 (2013).

\bibitem{Helstrom1976} Helstrom, C. W. 
\textit{Quantum Detection and Estimation Theory} (Academic Press, 1976).

\bibitem{Holevo1982} Holevo, A. S.
\textit{Probabilistic and Statistical Aspects of Quantum Theory} (North-Holland Publishing Company, 1982).

\bibitem{Braunstein1994} Braunstein, S. L. \& Caves, C. M.
Statistical distance and the geometry of quantum states.
\textit{Phys. Rev. Lett.} \textbf{72}, 3439--3443 (1994).

\bibitem{WisemanBook} Wiseman, H. M. \& Milburn,  G. J.
\textit{Quantum Measurement and Control} (Cambridge Univ. Press, 2009).

\bibitem{Paris2009} Paris, M. G. A.
Quantum estimation for quantum technology.
\textit{Int. J. Quant. Inf.} \textbf{7}, 125--137 (2009).

\bibitem{Provost1980} Provost, J. \&  Vallee, G.
Riemannian structure on manifolds of quantum states.
\textit{Commun. Math. Phys.} \textbf{76}, 289 (1980).

\bibitem{Shapere1989} Shapere, A. \& Wilczek, F. (Editors).
\textit{Geometric Phase in Physics} (World Scientific, 1989).

\bibitem{Personick1971} Personick, S. 
Application of quantum estimation theory to analog communication over quantum channels.
\textit{IEEE Trans. Inform. Theor.} \textbf{17}, 240--246 (1971).

\bibitem{Ozawa2003} Ozawa, M.
Universally valid reformulation of the Heisenberg uncertainty principle on noise and disturbance in measurement.
\textit{Phys. Rev. A} \textbf{67}, 042105 (2003).


\end{thebibliography}

\begin{thebibliography}{99}

\bibitem{Chefles1998} Chefles, A.
Unambiguous discrimination between linearly independent quantum states.
\textit{Phys. Lett. A} \textbf{239}, 339 (1998).

\bibitem{Dittmann1999} Dittmann, J.
Explicit formulae for the Bures metric.
\textit{J. Phys. A: Math. Gen.} \textbf{32}, 2663 (1999).

\end{thebibliography}
\end{document}